%% file: main.tex
\documentclass[a4paper,UKenglish]{lipics-v2019}
\usepackage[utf8]{inputenc}
\usepackage{comment}
\usepackage{algorithm}
\usepackage[noend]{algpseudocode}

\nolinenumbers

\let\originalleft\left
\let\originalright\right
\renewcommand{\left}{\mathopen{}\mathclose\bgroup\originalleft}
\renewcommand{\right}{\aftergroup\egroup\originalright}
\newcommand{\ket}[1]{\left| #1 \right\rangle}

\newcommand{\sign}{\operatorname{sign}}
\newcommand{\NULL}{\operatorname{NULL}}

\newcommand{\FindFrom}[1]{\textsc{FindFrom}_{#1}}
\newcommand{\FindFixedLen}[1]{\textsc{FindFixedLen}_{#1}}

\newcommand{\FindAny}[1]{\textsc{FindAny}_{#1}}
\newcommand{\FindFirst}[1]{\textsc{FindFirst}_{#1}}

\DeclareMathOperator{\dyck}{\textsc{Dyck}}

\theoremstyle{plain}

\title{Quantum Query Complexity of Dyck Languages with Bounded Height}

\author{Kamil Khadiev}{Kazan Federal University, Kazan, Russia}{kamilhadi@gmail.com}{}{}
\author{Yixin Shen}{Université de Paris, IRIF, CNRS, F-75013 Paris, France}{yixin.shen@irif.fr}{}{}

\authorrunning{K. Khadiev and Y. Shen}

\Copyright{Kamil Khadiev and Yixin Shen}

\ccsdesc[100]{Quantum query complexity}

\keywords{Dyck Languages with Bounded Height}

\category{}

\relatedversion{}

\supplement{}

\funding{Research supported in part by the ERA-NET Cofund in Quantum Technologies project QuantAlgo and the French ANR Blanc project RDAM.}

\acknowledgements{We thank Frédéric Magniez for introducing us to the problem and for helpful discussions.}

\EventEditors{John Q. Open and Joan R. Access}
\EventNoEds{2}
\EventLongTitle{42nd Conference on Very Important Topics (CVIT 2016)}
\EventShortTitle{CVIT 2016}
\EventAcronym{CVIT}
\EventYear{2016}
\EventDate{December 24--27, 2016}
\EventLocation{Little Whinging, United Kingdom}
\EventLogo{}
\SeriesVolume{42}
\ArticleNo{23}

\begin{document}

\maketitle

\begin{abstract}
We consider the problem of determining if a sequence of parentheses
is well parenthesized, with a depth of at most h. We denote this language as $Dyck_h$. We study the quantum query complexity of this problem for different h as function of the length n of the word.
It has been known from a recent paper by Aaronson et al. that, for any constant h, since $Dyck_h$ is star-free, it has quantum query complexity $\Tilde{\Theta}(\sqrt{n})$, where the hidden logarithm factors in $\Tilde{\Theta}$ depend on h. Their proof does not give rise to an algorithm.
When h is not a constant, $Dyck_h$ is not even context-free. We give an algorithm with $O\left(\sqrt{n}\log(n)^{0.5h}\right)$ quantum queries for $Dyck_h$ for all h. This is better than the trival upper bound $n$ when $h=o(\frac{\log(n)}{\log\log n})$. We also obtain lower bounds: we show that for every $0<\epsilon\leq 0.37$, there exists $c>0$ such that  $Q(\text{Dyck}_{c\log(n)}(n))=\Omega(n^{1-\epsilon})$.
When $h=\omega(\log(n))$, the quantum query complexity is close to $n$, i.e. $Q(\text{Dyck}_h(n))=\omega(n^{1-\epsilon})$ for all $\epsilon>0$. 
Furthermore when $h=\Omega(n^\epsilon)$ for some $\epsilon>0$, $Q(\text{Dyck}_{h}(n))=\Theta(n)$.
\end{abstract}

\section{Introduction}
Formal languages have a long history of study in classical theoretical computer science, starting with the study of regular languages back to Kleene in the 1950s \cite{Kleene56}. Roughly speaking, a formal language consists of an alphabet of letters, and a set of rules for generating words from those letters. Chomsky’s hierarchy is an early attempt to answer the following question: “Given more complex rules, what kinds of languages can we generate?”. The most well-known types of languages in that hierarchy are the regular and context-free languages. 
Modern computational complexity theory is still defined in terms of languages: complexity classes are defined as the sets of the formal languages that can be parsed by machines with certain computational powers.

The relationship between the Chomsky hierarchy and other models of computations has been studied extensively in many models, including Turing machines, probabilistic machines \cite{RABIN1963230}, quantum finite automata \cite{ambainis2015automata}, streaming algorithms \cite{Magniez,BABU201313} and query complexity \cite{Alon}. Query complexity is also known as the ‘black box model’, in this setting we only count the number of times that we need to query (i.e. access) the input in order to carry out our computation. It has been observed that quantum models of computation allow for significant improvements in the query complexity, when the quantum oracle access to the input bits is available \cite{Deutsch}. We assume the reader is familiar with the basis of quantum computing. One may refer to \cite{Nielsen} for a more detailed introduction to this topic.

The recent work by Scott Aaronson, Daniel Grier, and Luke Shaeffer \cite{Aaronson2018AQQ} is the first to study the relationship between the regular languages and quantum query complexity. They gives a full characterization of regular languages in the quantum query complexity model. More precisely, they show that every regular language naturally falls into one of three categories:

\begin{itemize}
\item ‘Trivial’ languages, for which membership can be decided by the first and last characters of the input string. For instance, the language describing all binary representations of even numbers is trivial.
\item Star-free languages, a variant of regular languages where complement is allowed ($\overline A$ — i.e. `something not in A’), but the Kleene star is not. The quantum query complexity of these languages is $\tilde{\Theta}(\sqrt{n})$.
\item All the rest, which have quantum query complexity $\Theta(n)$.
\end{itemize}

The proof uses the algebraic definitions of regular languages (i.e. in terms of monoids). Starting from an aperiodic monoid, Schützenberger constructs a star-free language recursively based on the “rank” of the monoid elements involved. \cite{Aaronson2018AQQ} uses this decomposition of star-free language of higher rank into star-free languages of smaller rank to show by induction that any star-free languages has $\tilde{\Theta}(\sqrt{n})$ quantum query complexity. However their proof does not immediately give rise to an algorithm.

One of the star-free language mentioned in \cite{Aaronson2018AQQ} is the Dyck language (with one type of parenthesis) with a constant bounded height.
The Dyck language is the set of balanced strings of brackets "(" and ")". When at any point the number of opening parentheses exceeds the number of closing parentheses by at most $h$, we denote the language as $Dyck_h$. 

The Dyck language is a fundamental example of a context-free language that is not regular. When more types of parenthesis are allowed, the famous Chomsky–Schützenberger representation theorem shows that any context-free language is the homomorphic image of the intersection of Dyck language and a regular language.

\paragraph*{Contributions}
We give an explicit algorithm (see Theorem~\ref{th:upper_bound}) for the decision problem of $Dyck_h$ with $O\left(\sqrt{n}\log(n)^{0.5h}\right)$ quantum queries. The algorithm also works when $h$ is not a constant and is better than the trival upper bound $n$ when $h=o(\frac{\log(n)}{\log\log n})$. We note that when $h$ is not a constant,
that is, if the height is allowed to depend on the length of the word,
$Dyck_h$ is not context-free anymore, therefore previous results do not apply.
We also obtain lower bounds on the quantum query complexity.
We show (Theorem~\ref{th:lower_bound_small_eps}) that for every $0<\epsilon\leq 0.37$, there exists $c>0$ such that  $Q(\text{Dyck}_{c\log(n)}(n))=\Omega(n^{1-\epsilon})$.
When $h=\omega(\log(n))$, the quantum query complexity is close to $n$, i.e. $Q(\text{Dyck}_h(n))=\omega(n^{1-\epsilon})$ for all $\epsilon>0$, see Theorem~\ref{th:lower_bound_omega_log}.
Furthermore when $h=\Omega(n^\epsilon)$ for some $\epsilon>0$, we show (Theorem~\ref{th:lower_bound_all}) that $Q(\text{Dyck}_{h}(n))=\Theta(n)$. Similar lower bounds were recently independently proven by Ambainis, Balodis, Iraids, Pr\=usis, and Smotrovs \cite{riga}, and Buhrman, Patro and Speelman \cite{buhrman2019quantum}.

\paragraph*{Structure of the paper}
In the next section, we give some definitions. In the following section we provide an algorithm of quantum query complexity $O(\sqrt{n}\log(n)^{0.5h})$ for $Dyck_h$. In the last section, we show some lower bounds when $h$ is $\Omega(\log(n))$.

\section{Definitions}
\label{s:defs}

For a word $x\in\Sigma^*$ and a symbol $a\in\Sigma$, let $|x|_a$ be the number of occurrences of $a$ in $x$.

For two (possibly partial) Boolean functions $g: G \rightarrow \{0,1\}$, where $G \subseteq \{0,1\}^n$, and $h:H\rightarrow \{0,1\}$, where $H\subseteq \{0,1\}^m$, we define the composed function $g\circ h: D\rightarrow \{0,1\}$, with $D\subseteq \{0,1\}^{nm}$, as
\[
\left(g\circ h\right) (x) = g\left(h(x_1,\dots,x_m), \dots, h(x_{(n-1)m+1},\dots, x_{nm})\right).
\]
Given a Boolean function $f$ and a nonnegative integer $d$, we define $f^d$ recursively as $f$ iterated $d$ times: $f^d=f\circ f^{d-1}$ with $f^1=f$.

{\bf Quantum query model.}
    We use the standard form of the quantum query model. 
    Let $f:D\rightarrow \{0,1\},D\subseteq \{0,1\}^n$ be an $n$ variable function we wish to compute on an input $x\in D$. We have an oracle access to the input $x$ --- it is realized by a specific unitary transformation usually defined as $\ket{i}\ket{z}\ket{w}\rightarrow \ket{i}\ket{z+x_i\pmod{2}}\ket{w}$ where the $\ket{i}$ register indicates the index of the variable we are querying, $\ket{z}$ is the output register, and $\ket{w}$ is some auxiliary work-space. An algorithm in the query model consists of alternating applications of arbitrary unitaries independent of the input and the query unitary, and a measurement in the end. The smallest number of queries for an algorithm that outputs $f(x)$ with probability $\geq \frac{2}{3}$ on all $x$ is called the quantum query complexity of the function $f$ and is denoted by $Q(f)$.
    
    Let a symmetric matrix $\Gamma$ be called an adversary matrix for $f$ if the rows and columns of $\Gamma$ are indexed by inputs $x\in D$ and $\Gamma_{xy}=0$ if $f(x)=f(y)$. Let $\Gamma^{(i)}$ be a similarly sized matrix such that $\Gamma^{(i)}_{xy}=\begin{cases}\Gamma_{xy}&\text{ if }x_i\neq y_i\\ 0&\text{ otherwise}\end{cases}$. Then let 
    \[Adv^{\pm}(f)=\max_{\Gamma\text{ - an adversary matrix for }f}{\frac{\|\Gamma\|}{\max_i{\|\Gamma^{(i)}\|}}}\]
    be called the adversary bound and let
    \[Adv(f)=\max_{\substack{\Gamma\text{ - an adversary matrix for }f\\ \Gamma \text{ - nonnegative}}}{\frac{\|\Gamma\|}{\max_i{\|\Gamma^{(i)}\|}}}\]
    be called the positive adversary bound.
        The following facts will be relevant for us:
    \begin{itemize}
        \item $Adv(f)\leq Adv^\pm(f)$;
        \item $Q(f)=\Theta(Adv^{\pm}(f))$ \cite{Reichardt11};
        \item $Adv^{\pm}$ composes exactly even for partial Boolean functions $f$ and $g$, meaning, $Adv^\pm(f\circ g)=Adv^\pm(f)\cdot Adv^\pm(g)$ \cite[Lemma~6]{kimmel2012quantum}
    \end{itemize}
    
{\bf Reductions.}
    We will say that a Boolean function $f$ is reducible to $g$ and denote it by $f \leqslant g$ if there exists an algorithm that given an oracle $O_x$ for an input of $f$ transforms it into an oracle $O_y$ for $g$ using at most $O(1)$ calls of oracle $O_x$ such that $f(x)$ can be computed from $g(y)$. Therefore, from $f \leqslant g$ we conclude that $Q(f)\leq Q(g)$ because one can compute $f(x)$ using the algorithm for $g(y)$ and the reduction algorithm that maps $x$ to $y$.

{\bf Dyck languages of bounded depth.}
Let $\Sigma$ be an alphabet consisting of two symbols: \texttt{(} and \texttt{)}.
The Dyck language $L$ consists of all $x\in \Sigma^*$ that represent a correct sequence of opening and closing parentheses.
We consider languages $L_k$ consisting of all words $x\in L$ where the number of opening parentheses that are not closed yet never exceeds $k$.

The language $L_k$ corresponds to a query problem $\dyck_{n, k}(x_1, ..., x_n)$ where $x_1, \ldots, x_n \in \{0, 1\}$ describe a word of length $n$
in the natural way: the $i^{\rm th}$ symbol of $x$ is \texttt{(} if $x_i=0$ and \texttt{)} if $x_i=1$. $\dyck_{n, k}(x)=1$ iff the word $x$ belongs to $L_k$. In the following, we sometimes use $\dyck_k(n)$ as a synonym of $\dyck_{n,k}()$. 

 For all $x \in \{0,1\}^n$, we define  $f(x)=|x|_{\texttt{0}}-|x|_{\texttt{1}}$, where $|x|_{a}$ is a number of $a$ symbols in $x$. We call $f$ the \textbf{balance}.
For all $0\leq i\leq n-1$, we define $x[i,j]=x_i,x_{i+1}\cdots x_j$. Finally, we define $h(x)= \max_{0\leq i\leq n-1}f(x[0,i])$ and $h^-(x)= \min_{0\leq i\leq n-1}f(x[0,i])$. We also define the function $\sign$ such that $\sign(a)=1$ if $a>0$, and $\sign(a)=-1$ if $a<0$, $\sign(a)=0$ if $a=0$. 

A substring $x[i,j]$ is called a $t$-substring if $f(x[i,j])=t$ for some integer $t$. A substring $x[i,j]$  is \emph{minimal} if it does not contain a substring $x[i',j']$ such that $(i,j)\neq (i',j')$, $f(x[i',j'])=f(x[i,j])$ and $i\geq i'\geq j'\geq j$.

\section{A quantum algorithm for membership testing of \texorpdfstring{$\dyck_{n,k}$}{DYCK\_\{n,k\}}}
In this section, we give a quantum algorithm for $\dyck_{n,k}(x)$, where $k$ can be a function of $n$. The general idea is that $\dyck_{n,k}(x)=0$ if and only if there are no $\pm(k+1)$-substrings in $1^nx0^n$. We thus first give an algorithm that searches for any $\pm k$-substrings and then give an algorithm for $\dyck_{n,k}$.

\subsection{\texorpdfstring{$\pm k$}{±k}-Substring Search algorithm}\label{sec:substr}


\input{k-substring-search}


\subsection{The Algorithm for \texorpdfstring{$\dyck_{n,k}$}{DYCK\_\{n,k\}}}

To solve $\dyck_{n,k}$, we modify the input $x$. As the new input we use $x'=1^k x  0^k$. $\dyck_{n,k}(x)=1$ iff there are no $\pm(k+1)$-substrings in $x'$. This idea is presented in Algorithm \ref{alg:dycknh}. 
\begin{algorithm}[h]
    \caption{$\textsc{Dyck}_{n,k}()$. The Quantum Algorithm for $\dyck_{n,k}$.\label{alg:dycknh}}
    \begin{algorithmic}
        \State $x\gets 1^k  x  0^k$
        \State $v=\FindAny{(k+1)}(0,n-1, \{+1,-1\})$
        \If{$v=\NULL$}
            \State \Return $1$
        \EndIf
        \If{$v\neq\NULL$}
            \State \Return $0$
        \EndIf
    \end{algorithmic}
\end{algorithm}

\begin{theorem}\label{th:upper_bound} Algorithm \ref{alg:dycknh} solves $\dyck_{n,k}$  and
the expected running time of Algorithm \ref{alg:dycknh} is $O(\sqrt{n}(\log n)^{0.5k})$. The algorithm has two-side error probability  $\varepsilon<0.5$.
\end{theorem}
\begin{proof}
Let us show that if $x'$ contains $\pm(k+1)$-substring then one of  three conditions of $\dyck_{n,k}$ problem is broken.

Assume that $x'$ contains $(k+1)$ substring $x'[i,j]$. If $j\geq k+n$, then $f(x[i-k,n-1])>0$, because $f(x'[n,j])=j-n+1\leq k<k+1$. Therefore, prefix $x[0,i-k]$ is such that $f(x[0,i-k-1])<0$ or $f(x[0,n-1])>0$ because  $f(x[0,n-1])=f(x[0,i-k])+f(x[i-k-1,n-1])$. So, in that case we break one of conditions of $\dyck_{n,k}$ problem.

If $j<k+n$ then $x[i-k,j-k]$ is $(k+1)$ substring of $x$. 

Assume that $x'$ contains $-(k+1)$ substring $x'[i,j]$. If $i< k$, then $f(x[0,j-k])<0$, because $f(x'[i,k-1])=-(k-i)\geq -k>-(k+1)$ and $f(x[0,j-k])= f(x'[k,j])=f(x[i,j])-f(x[i,k-1])$. So, in that case the second condition of $\dyck_{n,k}$ problem is broken.

The complexity of Algorithm \ref{alg:dycknh} is the same as the complexity of $\FindAny{k+1}$ for $x'$ that is $O(\sqrt{n+2k}(\log(n+2k))^{0.5k})$ due to Proposition \ref{pr:findfrom}.

We can assume $n\geq 2k$ (otherwise, we can update $k\gets n/2$). Hence,
\[O(\sqrt{n+2k}(\log(n+2k))^{0.5k})=
O(\sqrt{2n}(\log(2n))^{0.5k})=O(\sqrt{n}(2\log{n})^{0.5k})=O(\sqrt{n}(\log{n})^{0.5k})\]

The error probability is the same as the complexity of $\FindAny{k+1}$. 
\end{proof}

\section{Lower Bounds for Dyck Languages with Bounded Height}

Now let's show some lower bounds for Dyck languages with bounded height.

Let $k\in \mathbb{N^+}$. Let $M^0_k=\{a,b\}$. For all $i \in \mathbb{N^+}$, let $M^i_k=\{a^kwb^k|w\in (M^{i-1}_{k})^{2k-1},f(w)=\pm 1\}$. Here $f$ is the balance function that we defined in the notation paragraph.

According to our construction, for all $i\in \mathbb{N}$, for all $m \in M^i_k$, we have $f(m)=\pm 1$. All words in $M^i_k$ have the same length that we define as $l_k(i)$.

\begin{lemma}
$l_k(i)\sim 2\cdot(2k)^i$ as $k\to\infty$.
\end{lemma}

\begin{proof}

We have $l_k(0)=1$. For all $i \in \mathbb{N^+}$, $l_k(i)=2k+(2k-1)l_k(i-1)$.

Thus, $l_k(i)=(2k-1)^{i-1}((4k-1)+\frac{k}{k-1})-\frac{k}{k-1}\sim 2\cdot(2k)^i$ as $k\to\infty$.

\end{proof}

Define $h_k(i)=\max_{m\in M_k(i)}(h(m))$.

\begin{lemma}
For all $k\in \mathbb{N}^+, i\in \mathbb{N}^+$, $h_k(i)=(i+1)k$. Furthermore, for all $m\in M_k^i$, $h^-(m)\geq 0$.
\end{lemma}

\begin{proof}
This can be shown easily by induction on $i$.
\end{proof}

Define $$g_k: \{-1,1\}^{2k-1} \supset C \longrightarrow \{-1,1\}$$
$$(x_1,x_2,\cdots,x_{2k-1}) \mapsto \sum_{i=1}^{2k-1}x_i.$$
By induction, define $g_k^1=g_k$ and $g_k^{i+1}=g_k\circ(\overbrace{g_k^i,\cdots,g_k^i}^{2k-1\ \text{times}})$.

\begin{lemma}
    \label{lemAdv}
    For all $k\in \mathbb{N}^+$, $i\in \mathbb{N}^+$, $ADV^{\pm}(g_k^i)\geq k^i$. 
\end{lemma}

\begin{proof}
Inspiring from \cite{b2014} Prop 3.32, we can show that $ADV^{\pm}(g_k)\geq k$. 

Since $ADV^{\pm}$ composes exactly even for partial Boolean functions $f$ and $g$, meaning,
$Adv^{\pm}(f \circ g) = Adv{\pm}(f)\cdot Adv^{\pm}(g)$ \cite[Lemma 6]{Kimmel}, we have $ADV^{\pm}(g_k^i)\geq k^i$.
\end{proof}
\smallskip

The reason that we introduced $g_k^i$ is the following:
Let $m \in M^i_k$, $mb \in \text{Dyck}_{(i+1)k}(l(i))$ iif $g_k^i(m')=1$, where $m'$ is obtained from $m$ by removing all the $a^k$s' and $b^k$s' appeared in the construction of each $M_k^j$ for $j=1$ to $i$. 
\medskip

Now let's study the lower bound of Dyck language with bounded height.

\begin{theorem}\label{th:lower_bound_all}
For all $\epsilon>0$, $Q(\text{Dyck}_{\Omega(n^\epsilon)}(n))=\Theta(n)$
\end{theorem}

\begin{proof}
We know that $l_k(i)\sim 2\cdot(2k)^i$. $h_k(i)=(i+1)k$. From \cite{Reichardt}, we have $Q(g^i)=ADV^\pm(g_k^i)$. Thus Lemma \ref{lemAdv} shows that $Q(g^i)\geq k^i$. 
\smallskip

By taking $i=constant$, $k=\Theta(n^{1/i})$, we have $l_k(i)=n$ and $Q(g^i)=\Theta(n)$. Furthermore, by the equivalence above, computing $g_k^i$ corresponds to checking
if words of height $\Theta(n^{1/i})$ are in Dyck. Thus $Q(\text{Dyck}_{\Theta(n^{1/i})}(n))=\Theta(n)$. This is true for all $i \in \mathbb{N}^+$. Therefore, for all $\epsilon>0$, $Q(\text{Dyck}_{\Omega(n^\epsilon)}(n))=\Theta(n)$.
\end{proof}

\begin{theorem}\label{th:lower_bound_omega_log}
  $Q(\text{Dyck}_{\Theta(i\cdot n^{1/i})}(n))=\Omega(n/2^i)$ for $i=i(n)$, such that $i(n)\in [\omega(1),o(\log n)]$ as $n \to \infty$.
\end{theorem}

\begin{proof}
We know that $l_k(i)\sim 2\cdot(2k)^i$ when $k=k(n)=\omega(1)$. $h_k(i)=(i+1)k\sim ik$ when $i=\omega(1)$. $Q(g^i)=ADV^\pm(g_k^i)\geq k^i$.
\smallskip

By replacing $k$ by $\Theta(n^{1/i})$, we obtain $Q(\text{Dyck}_{\Theta(i\cdot n^{1/i})}(n))=\Omega(n/2^i)$.
\end{proof}

\begin{theorem}\label{th:lower_bound_small_eps}
 For every $0<\epsilon\leq 1-\log_3(2)\approx 0.37$, there exists $c>0$ such that  $Q(\text{Dyck}_{c\log(n)}(n))=\Omega(n^{1-\epsilon})$.
\end{theorem}

\begin{proof}
We know that $l_k(i)\sim 2\cdot(2k-1)^i$, $h_k(i)=(i+1)k\sim ik$ when $i(n)\rightarrow \infty $ and $k$ equals to a constant, $k>1$. $Q(g^i)=ADV^\pm(g_k^i)\geq k^i$.
\smallskip

By taking $i$ as $\Theta(\log(n))$, we obtain $h=c\log(n)$ for some $c>0$, $k^i=2(2k-1)^{i(1-\epsilon)}$ for $\epsilon=1-\log_{2k-1}(k)$.
Since $k$ is an integer, $\log_{2k-1}(k)\leq \log_3(2)$ when $k\geq 2$.

For every $0<\epsilon\leq 1-\log_3(2)\approx 0.37$, there exists $c>0$ such that  $Q(\text{Dyck}_{c\log(n)}(n))=\Omega(n^{1-\epsilon})$.
\end{proof}

\bibliographystyle{plainurl}
\bibliography{references}

\appendix
\section{\texorpdfstring{$\pm 2$}{±2}-Substring Search Algorithm}\label{sec:substr2}
\input{2-substring-search}

\end{document}

%% file: k-substring-search.tex
%
The goal of this section is to describe a quantum algorithm which searches for a substring $x[i,j]$ that has a balance $f(x[i,j])\in\{+k,-k\}$ for some integer $k$. 

Throughout this section, we find and consider only minimal substrings. For any two [minimal] $\pm k$-substrings $x[i,j]$ and $x[k,l]$: $i<k \implies j<l$. This induces a natural linear order among all $\pm k$-substrings according to their starting (or, equivalently, ending) positions. Furthermore, substrings of opposite signs do not intersect at all.

This algorithm is the basis of our algorithms for $\dyck_{n,k}$. The easiest case when $k=2$ is shown in Appendix~\ref{sec:substr2}. The main idea for $k=2$ is to use Grover's search algorithm to search for two sequential equal symbols.

The algorithm works in a recursive way. It searches for two $\pm (k-1)$-substrings $x[l_1,r_1]$ and $x[l_2,r_2]$  such that there are no $\pm (k-1)$-substrings between them. If both substrings $x[l_1,r_1]$ and $x[l_2,r_2]$ are $+(k-1)$-substrings, then we get a $+k$-substring in total. If both substrings are $-(k-1)$-substrings, then we get a $-k$-substring in total.

We first discuss two building blocks for our algorithm. The first one is $\FindFrom{k}(l,r,t,d,s) $ and accepts as inputs:
\begin{itemize}
    \item the borders $l$ and $r$, where $l$ and $r$ are integers such that $0\leq l \leq r\leq n-1$;
    \item a position $t\in \{l,\dots, r\}$;
    \item a maximal length $d$ for the substring, where $d$ is an integer such that $0<d\leq r-l+1$;
    \item the sign of the balance $s\subseteq \{+1,-1\}$. $+1$ is used for searching for a $+k$-substring, $-1$ is used for searching for a $-k$-substring, $\{+1,-1\}$ is used for searching for both.  
\end{itemize}

It outputs a triple $(i,j,\sigma)$ such that $t\in[i,j]$, $j-i+1\leq d$, $f(x[i,j])\in\{+k,-k\}$ and $\sigma=\sign(f(x[i,j]))\in s$. If no such substrings have been found, the algorithm returns $\NULL$.

The second one is $\FindFirst{k}(l,r,s,direction)$ and accepts as inputs:
\begin{itemize}
    \item the borders $l$ and $r$, where $l$ and $r$ are integers such that $0\leq l \leq r\leq n-1$;
    \item the sign of the balance $s\subseteq \{+1,-1\}$. $+1$ is used for searching for a $+k$-substring, $-1$ is used for searching for a $-k$-substring, $\{+1,-1\}$ is used for both;
    \item a $direction \in \{left,right\}$.
\end{itemize} 
It outputs a triple $(i,j,\sigma)$ such that $l\leq i\leq j\leq r$, $f(x[i,j])\in\{+k,-k\}$ and $\sigma=\sign(f(x[i,j]))\in s$. Furtherfore, if the direction is "right", then $x[i,j]$ is the first substring starting from the index $l$ to the right that satisfies all previous conditions. If the direction is "left", then $x[i,j]$ is the first substring starting from the index $r$ to the left that satisfies all previous conditions. The algorithm returns $\NULL$, if it cannot find such a substring.

These two building blocks are interdependent since $\FindFrom{k}$ uses $\FindFirst{k-1}$ as a subroutine and $\FindFirst{k}$ uses $\FindFrom{k}$ as a subroutine. 
A description of implementation of $\FindFrom{k}(l,r,t,d,s)$  follows. The algorithm is presented in Appendix \ref{sec:desc_k_substring_base}.
\begin{description}
    \item[Step $1$.] We check whether $t$ is inside a $\pm (k-1)$-substring of length at most $d-1$, i.e. 
        
        $v=(i,j,\sigma)\gets  \FindFrom{k-1}(l,r,t,d-1,\{+1,-1\}).$
        
        If $v\neq \NULL$, then $(i_1,j_1,\sigma_1)\gets(i,j,\sigma)$ and the algorithm goes to Step $2$. Otherwise, the algorithm goes to Step $6$.
    \item[Step $2$.] We check whether $i_1-1$ is inside a $\pm (k-1)$-substring of length at most $d-1$, i.e. 
    
    $v=(i,j,\sigma)\gets\FindFrom{k-1}(l,r,i_1-1,d-1,\{+1,-1\}).$
    
        If $v=\NULL$, then the algorithm goes to Step $3$.
        If $v\neq \NULL$ and $\sigma=\sigma_1$, then $(i_2,j_2,\sigma_2)\gets(i,j,\sigma)$ and the algorithm goes to Step $8$. Otherwise, the algorithm goes to Step $4$.
    \item[Step $3$.] We search for the first $\pm (k-1)$-substring on the left from $i_1-1$ at distance at most $d$, i.e. 
    
    $v=(i,j,\sigma)\gets\FindFirst{k-1}(\min(l,j_1-d+1),i_1-1),\{+1,-1\},left).$
    
    If $v\neq \NULL$ and $\sigma_1=\sigma$, then $(i_2,j_2,\sigma_2)\gets(i,j,\sigma)$ and the algorithm goes to Step $8$. Otherwise, the algorithm goes to Step $4$.    
        
    \item[Step $4$.] We check whether $j_1+1$ is inside a $\pm (k-1)$-substring of length at most $d-1$, i.e. 
    
    $v=(i,j,\sigma)\gets\FindFrom{k-1}(l,r,j_1+1,d-1,\{+1,-1\}).$
    
        If $v\neq \NULL$, then $(i_2,j_2,\sigma_2)\gets(i,j,\sigma)$ and the algorithm goes to Step $8$. Otherwise, the algorithm goes to Step $5$.
        
    \item[Step $5$.] We search for the first $\pm (k-1)$-substring on the right from $j_1+1$ at distance at most $d$, i.e. 
    
    $v=(i,j,\sigma)\gets\FindFirst{k-1}(j_1+1,\min(i_1+d-1,r),\{+1,-1\},right).$ 
    
    If $v\neq \NULL$, then $(i_2,j_2,\sigma_2)\gets(i,j,\sigma)$. The algorithm goes to Step $8$. Otherwise, the algorithm fails and returns $\NULL$.

    \item[Step $6$.] We search for the first $\pm (k-1)$-substring on the right at distance at most $d$ from $t$, i.e.
    
   $v=(i,j,\sigma)\gets\FindFirst{k-1}(t,\min(t+d-1,r),\{+1,-1\},right)$
   
    If $v\neq \NULL$, then $(i_1,j_1,\sigma_1)\gets(i,j,\sigma)$ and the algorithm goes to Step $7$. Otherwise, the algorithm fails and returns $\NULL$.
    \item[Step $7$.] We search for the first $\pm (k-1)$-substring on the left from $t$ at distance at most $d$, i.e. 
    
    $v=(i,j,\sigma)\gets\FindFirst{k-1}(\max(l,t-d+1),t),\{+1,-1\},left)$
    
    If  $v\neq \NULL$, then $(i_2,j_2,\sigma_2)\gets(i,j,\sigma)$ and go to Step $8$. Otherwise, the algorithm fails and returns $\NULL$.
    \item[Step $8$.] If $\sigma_1=\sigma_2$, $\sigma_1\in s$ and $\max(j_1,j_2)-\min(i_1,i_2)+1\leq d$ , then we output $[\min(i_1,i_2),\max(j_1,j_2)]$, otherwise the algorithm fails and returns $\NULL$. 
\end{description}

\begin{algorithm}[h]
    \caption{$\FindFrom{k}(l,r,t,d,s)$. Search any $\pm k$-substring.\label{alg:substringk_base}}
    \begin{algorithmic}
    \State $v=(i_1,j_1,\sigma_1)\leftarrow \FindFrom{k-1}(l,r,t,d-1,\{+1,-1\})$
    \If{$v \neq \NULL$}\Comment{if $t$ is inside a $\pm(k-1)$-substring}
        \State $v'=(i_2,j_2,\sigma_2)\leftarrow \FindFrom{k-1}(l,r,i_1-1,d-1,\{+1,-1\})$
        \If{$v'=\NULL$}
            \State $v'=(i_2,j_2,\sigma_2)\gets\FindFirst{k-1}(\min(l,j_1-d+1),i_1-1),\{+1,-1\},left)$
        \EndIf
        \If {$v'\neq\NULL$ and $\sigma_2\neq\sigma_1$}
        \State $v'\gets\NULL$
        \EndIf
        \If{$v'=\NULL$}
            \State $v'=(i_2,j_2,\sigma_2)\leftarrow \FindFrom{k-1}(l,r,j_1+1,d-1,\{+1,-1\})$
        
        \If{$v'=\NULL$}
            \State $v'=(i_2,j_2,\sigma_2)\leftarrow \FindFirst{k-1}(j_1+1,\min(i_1+d-1,r),\{+1,-1\},right)$
        \EndIf
        \EndIf
        \If{$v'=\NULL$}
            \State \Return{$\NULL$}
        \EndIf
    \Else
        \State $v=(i_1,j_1,\sigma_1)\gets\FindFirst{k-1}(t,\min(t+d-1,r),\{+1,-1\},right)$
        \If{$v=\NULL$}
            \State \Return{$\NULL$}
        \EndIf
        \State $v'=(i_2,j_2,\sigma_2)\gets\FindFirst{k-1}(\max(l,t-d+1),t),\{+1,-1\},left)$
        \If{$v'=\NULL$}
            \State \Return{$\NULL$}
        \EndIf
    \EndIf
    \If{$\sigma_1=\sigma_2$ and $\sigma\in s$ and $\max(j_1,j_2)-\min(i_1,i_2)+1\leq d$}
        \State\Return{$(\min(i_1,i_2),\max(j_1,j_2),\sigma_1)$}
    \Else
        \State \Return{$\NULL$}
    \EndIf
    \end{algorithmic}
\end{algorithm}

Using this basic procedure, we then search for a $\pm k-$substring by searching for a $t$ and $d$ such that
$\FindFrom{k}(l,r,t,d,s)$
returns a non-$\NULL$ value. Unfortunately,
our algorithms have two-sided bounded error: they
can, with small probability, return $\NULL$ even if a substring exists or return a wrong substring instead
of $\NULL$. In this setting, Grover's search algorithm
is not directly applicable and we need to use a more
sophisticated search~\cite{Hoyer03}. Furthermore, simply applying
the search algorithm naively does not give the right complexity.
Indeed, if we search for a substring of length roughly $d$
(say between $d$ and $2d$), we can find one with expected running time $O\left(\sqrt{\tfrac{r-l}{d}}\right)$
because at least $d$ values of $t$ will work.
On the other hand, if there are no such substrings, the expected running time will be 
$O(\sqrt{r-l})$. Intuitively, we can do better
because if there is a substring of length at least $d$ then there are at least $d$ values of $t$ that work. Hence, we only need to distinguish between no solutions, or at least $d$. This allows to stop the Grover iteration early and make
$O\left(\sqrt{\tfrac{r-l}{d}}\right)$ queries in all cases.

\begin{lemma}[Modified from \cite{Hoyer03}]\label{lem:modified_hoyer03}
    Given $n$ algorithms, quantum or classical, each computing some bit-value with bounded error probability, and some $T\geqslant1$,
    there is a quantum algorithm that uses $O(\sqrt{n/T})$
    queries and with constant probability:
    \begin{itemize}
        \item returns the index of a ``1'', if they are at least $T$
            ``1s'' among the $n$ values,
        \item returns $\NULL$ if they are no ``1'',
        \item returns anything otherwise.
    \end{itemize}
\end{lemma}
\begin{proof}
    The main loop of the algorithm of \cite{Hoyer03} is the following,
    assuming the algorithms have error at most $1/9$:
    \begin{itemize}
        \item for $m=0$ to $\lceil\log_9 n\rceil$-1 do:
            \begin{enumerate}
                \item run $A_m$ 1000 times,
                \item verify the 1000 measurements, each by $O(\log n)$ runs of the corresponding algorithm,
                \item if a solution has been found, then output a solution and stop
            \end{enumerate}
        \item Output ‘no solutions’
    \end{itemize}
    The key of the analysis is that if the (unknown) number $t$
    of solutions lies in the interval $[n/9^{m+1},n/9^m]$, then
    $A_m$ succeeds with constant probability. In all cases, if they are
    no solutions, $A_m$ will never succeeds with high probability (ie the algorithm only applies good solutions).
    
    In our case, we allow the algorithm to return anything (including
    $\NULL$) if $t<T$. This means that we only care about the values of
    $m$ such that $n/9^{m}\geqslant T$, that is
    $m\leqslant\log_9\tfrac{n}{T}$. Hence, we simply run the algorithm
    with this new upper bound for $d$ and it will satisfy our
    requirements with constant probability. The complexity is
    
    $
        \sum\limits_{m=0}^{\left\lfloor log_9\tfrac{n}{T}\right\rfloor}1000\cdot O(3^m)+1000\cdot O(\log n)=O(3^{\log_9\tfrac{n}{T}})=O(\sqrt{n/T}).
    $

\end{proof}

The algorithm that uses above ideas is presented in Algorithm \ref{alg:fixedlenk}.

\begin{algorithm}[ht]
    \caption{$\FindFixedLen{k}(l,r,d,s)$. Search for any $\pm k$-substring of length $\in[d/2,d]$\label{alg:fixedlenk}}
    \begin{algorithmic}
        \State Find $t$ such that $v_t\gets\FindFrom{k}(l,r,t,d,s)\neq\NULL$ using Lemma~\ref{lem:modified_hoyer03} with $T=d/2$.
        \State \Return{$v_t$ or $\NULL$ if none}.
    \end{algorithmic}
\end{algorithm}


We can then write an algorithm $\FindAny{k}(l,r,s)$ that searches for any $\pm k$-substring. We can consider a randomized algorithm that uniformly chooses on of power $2$ from $[2^{\lceil\log_2 k\rceil},  (r-l)]$ segment, i.e.  $d\in\{2^{\lceil\log_2 k\rceil}, 2^{\lceil\log_2 k\rceil+1},\dots, 2^{\lceil\log_2 (r-l)\rceil}\}$. For the chosen $d$, we run Algorithm \ref{alg:fixedlenk}. So, the algorithm will successful with probability at least $O(1/\log (r-l))$. We can apply Amplitude amplification and ideas from Lemma \ref{lem:modified_hoyer03} to this algorithm and get the algorithm that uses $O(\sqrt{\log n})$ iterations. 

\begin{algorithm}[ht]
    \caption{$\FindAny{k}(l,r,s)$. Search for any $\pm k$-substring.\label{alg:substringk_any}}
    \begin{algorithmic}
        \State Find $d\in\{2^{\lceil\log_2 k\rceil}, 2^{\lceil\log_2 k\rceil+1},\dots, 2^{\lceil\log_2 (r-l)\rceil}\}$ such that $v_d\gets\FindFixedLen{k}(l,r,d,s)\neq\NULL$ using amplitude amplification.
        \State \Return{$v_d$ or $\NULL$ if none}.
    \end{algorithmic}
\end{algorithm}

Finally, we can write an algorithm that finds the first $\pm k$-substring. The idea is similar to the first one search algorithm from \cite{k2014,ll2016}. We search for a $\pm k$-substring in the segment of length $w$ that is a power of $2$. Assume that the answer is $x[i,j]$ and we search it on the left in the segment $[l,r]$, then the first time when we find the substring is the case $v=2^{\lceil\log_2 (l-j) \rceil}\leq 2(l-j)$.
Procedure $\FindFirst{k}$ in Algorithm \ref{alg:substringk_first}
implements this idea.

\begin{algorithm}[ht]
    \caption{$\FindFirst{k}(l,r,s,direction)$. The algorithm for searching for the first $\pm k$-substring.\label{alg:substringk_first}}
    \begin{algorithmic}
        \State $w\gets 2^{\lceil \log_2 k \rceil}$
        \State $v\gets NULL$
        \While{$w\leq 2(r-l)$ and $v=NULL$}
        \If{$direction=left$}
        \State $v = (i,j,\sigma) \gets \FindAny{k}(l,\min(r,l+w-1),s)$
        \EndIf
        \If{$direction=right$}
         \State $v = (i,j,\sigma) \gets \FindAny{k}(\max(l,r-w+1),r,s)$
        \EndIf
        \State $w \gets w\cdot 2$
        \EndWhile
               \If{$v\neq NULL$} 
               \State $v'\gets v$
                \While{$v'\neq NULL$}
                \State $v\gets v'$
               \If{$direction=left$}
                \State $v' = (i,j,\sigma) \gets \FindAny{k}(l,j-1,s)$
                \EndIf 
                \If{$direction=right$}
                \State $v' = (i,j,\sigma) \gets \FindAny{k}(i+1,r,s)$
                \EndIf
                \EndWhile

            \EndIf
            \State\Return{$v$}
    \end{algorithmic}
\end{algorithm}

\begin{proposition}\label{pr:findfrom}
    For any $\varepsilon>0$ and $k$, algorithms $\FindFrom{k}$, $\FindFixedLen{k}$, $\FindAny{k}$ and $\FindFirst{k}$ have two-sided error probability $\varepsilon<0.5$ and return, when correct:
    \begin{itemize}
        \item If $t$ is inside a $\pm k-$substring of sign $s$ of length
            at most $d$ in $x[l,r]$, then
            $\FindFrom{k}$
            will return such a substring, otherwise it returns $\NULL$.
            The expected running time is $O(\sqrt{d}(\log n)^{0.5(k-2)})$.
        \item $\FindFixedLen{k}$ either returns a
            $\pm k-$substring of sign $s$ and length at most $d$ in $x[l,r]$, or $\NULL$. It is only guaranteed to return a substring if there exists $\pm k-$substring of length at least $d/2$, otherwise it can return $\NULL$.
            The expected running time is $O(\sqrt{r-l}(\log (r-l))^{0.5(k-2)})$.
        \item $\FindAny{k}$ returns any
            $\pm k-$substring of sign $s$ in $x[l,r]$, otherwise it returns $\NULL$.
            The expected running time $O(\sqrt{r-l}(\log (r-l))^{0.5(k-1)})$.
        \item $\FindFirst{k}$ returns the first
            $\pm k-$substring in $x[i,j]$ in the specified direction, otherwise it returns $\NULL$.
            The expected running is $O(\sqrt{r-l}(\log(r-l))^{0.5(k-1)})$ if there are no such substrings.
    \end{itemize}
\end{proposition}

\begin{proof}
We prove the result by induction on $k$. The base case of $k=2$ is
in Appendix~\ref{sec:substr2}.
We first prove the correctness of all the algorithms, assuming there are
no errors. At the end we  explain how to deal with the errors.

\textbf{We start with $\FindFrom{k}$:}
there are different cases to be considered when searching for a $+k$-substring $x[i,j]$.
\begin{enumerate}
    \item Assume that there are $j_1$ and $i_2$ such that $i<j_1<i_2<j$, $|f(x[i,j_1])|=|f(x[i_2,j])|=k-1$ and $sign(f(x[i,j_1]))=sign(f(x[i_2,j]))\in s$.
    \begin{itemize}
        \item If $t\in\{i_2,\dots,j\}$, then the algorithm finds $x[i_2,j]$ on Step $1$ and the first invocation of $\FindFirst{k-1}$ on Step $3$  finds $x[i,j_1]$ in the case of $g=j-i+1\leq d$.
        \item If $t\in\{i,\dots,j_1\}$, then the algorithm finds $x[i,j_1]$ on Step $1$ and the  second invocation of $\FindFirst{k-1}$ on Step $5$ finds $x[i_2,j]$ in the case of $g=j-i+1\leq d$.   
        \item If $j_1<t<i_2$, then the third invocation of $\FindFirst{k-1}$ on Step $6$  finds $x[i_2,j]$ and the  forth invocation of $\FindFirst{k-1}$ on Step $7$  finds $x[i,j_1]$ in the case of $g=j-i+1\leq d$.
    \end{itemize}
    \item Assume that there are $j_1$ and $i_2$ such that $i<i_2<j_1<j$, $|f(x[i,j_1])|=|f(x[i_2,j])|=k-1$ and $sign(f(x[i,j_1]))=sign(f(x[i_2,j]))\in s$.
    \begin{itemize}
       \item If $t\in\{i_2,\dots,j\}$, then the algorithm finds $x[i_2,j]$ on Step $1$. After that, it finds $x[i,j_1]$ on Step $2$ in the case of $g=j-i+1\leq d$.
        \item If $t\in\{i,\dots,i_2-1\}$, then the algorithm finds $x[i,j_1]$ on Step $1$. After that, it finds $x[i_2,j]$ on Step $4$ in the case of $g=j-i+1\leq d$.  
    \end{itemize}
\end{enumerate}
By induction, the running time of each $\FindFrom{k-1}$ invocation is $O(\sqrt{d}(\log n)^{0.5(k-3)})$, and the running time of each $\FindFirst{k-1}$ invocation is $O(\sqrt{d}(\log n)^{0.5(k-2)})$.

\textbf{We now look at $\FindFixedLen{k}$:} by
construction and definition of $\FindFrom{k}$, if the algorithm
returns a value, it is a valid substring (with high probability).
If there exists a substring of length at least $d/2$, then any
query to $\FindFrom{k}$ with a value of $t$ in this interval will 
succeed, hence there are at least $d/2$ solutions. Therefore,
by Lemma~\ref{lem:modified_hoyer03}, the algorithm will find one
with high probability and make $O\left(\sqrt{\tfrac{r-l}{d/2}}\right)$ queries. 
Each query has complexity $O(\sqrt{d}(\log n)^{0.5(k-2)})$
by the previous paragraph, hence the running time is bounded by
$O(\sqrt{r-l}(\log n)^{0.5(k-2)})$.

\textbf{We can now analyze $\FindAny{k}$:}  Assume that the shortest $\pm k$-substring $x[i,j]$ is of length $g=j-i+1$. Therefore, there is $d$ from the search space such that $d\leq g \leq 2d$ and the $\FindFixedLen{k}$ procedure  returns the substring for this $d$ with constant success probability. So, the success probability of the randomized algorithm is at least $O(1/\log (l-r))$. Therefore, the amplitude amplification does $O(\sqrt{\log(r-l)})$ iterations. 
The running time of $\FindFixedLen{k}$ is $O(\sqrt{r-l}(\log n)^{0.5(k-2)})$ 
by induction, hence
the total running time is 
$O(\sqrt{r-l}(\log n)^{0.5(k-2)}\sqrt{\log (l-r)})=O(\sqrt{r-l}(\log n)^{0.5(k-1)})$.

\textbf{Finally, we analyze $\FindFirst{k}$:}
we can show the properties almost immediately from the proofs in \cite{Durr96aquantum,ll2016,k2014}.

\textbf{We now turn to the analysis of the errors.}
The case of $\FindFrom{k}$ is easy:
the algorithm makes at most $5$ recursive calls, each having a success
probability of $1-\varepsilon$. Hence it will succeed with probability
$(1-\varepsilon)^5$. We can boost this probability to $1-\varepsilon$
by repeating this algorithm a constant number of times. Note that this constant depends on $\varepsilon$.

The analysis of $\FindFixedLen{k}$ follows directly from \cite{Hoyer03}
and Lemma~\ref{lem:modified_hoyer03}:
since $\FindFrom{k}$ has two-sided error $\varepsilon$, there exists
a search algorithm with two-sided error $\varepsilon$.

\end{proof}

%% file: 2-substring-search.tex

The simplest case is an algorithm that searches for a substring $x[i,j]$ that has a balance $f(x[i,j])\in\{+2,-2\}$.
The algorithm looks for two sequential equal symbols using Grover's Search Algorithm \cite{g96,bbht98}. Formally, it is a procedure that accepts the following parameters as inputs and outputs:
\begin{itemize}
    \item {\bf Inputs:}
    \begin{itemize}
        \item an integer $l\in\{0,\dots,n-1\}$ which is a left border for the substring to be searched,
        \item an integer $r\in\{l,\dots,n-1\}$ which is a right border for the substring to be searched, 
        \item a set $s\subseteq \{+1,-1\}$ which represents the sign of the balance $f$ that we are looking for.
    \end{itemize}
    \item {\bf Outputs:}
    \begin{itemize}
        \item a triple $(i,j,\sigma)$ where $i$ and $j$ are the left and right border of the found substring and where $\sigma$ is the sign of $f(x[i,j])$, i.e. $\sigma=\sign(f(x[i,j]))$. If there are no such substrings, then the algorithm returns $\NULL$. Furthermore, when there is a satisfying substring, the result is such that $l\leq i\leq j\leq r$.
    \end{itemize}
\end{itemize}
The algorithm searches for a substring $x[i,j]$ such that $f(x[i,j])\in\{+2,-2\}$ and $\sigma=\sign(f(x[i,j]))\in s$.

\bigskip

We use Grover's Search Algorithm as a subroutine $\textsc{Grover}(l,r,{\cal F})$ that takes as inputs $l$ and $r$ as left and right borders of the search space and some function ${\cal F}:\{l,\dots,r\}\to\{0,1\}$. We search for any index $i$, where $l\leq i\leq r$, such that ${\cal F}(i)=1$ .
The result of the function $\textsc{Grover}(l,r,{\cal F})$ is either some index $i$ or $-1$ if it has not found the required $i$.
In Algorithm \ref{alg:substring2_base}, we use Grover's search on the function $\mathcal{F}_s:\{0,\dots,n-1\}\to\{0,1\}$ defined by
\[
    \mathcal{F}_s(i)=1 \quad\Leftrightarrow\quad (x_{i}=x_{i+1} \text{ or } x_i=x_{i-1})\text{ and the following conditions hold:}
\]
\begin{itemize}
    \item if $s=\{+1\}$ then $x_i=0$.
    \item if $s=\{-1\}$ then $x_i=1$.
    \item if $s=\{+1,-1\}$ then $x_i=0$ or $x_i=1$.
    \end{itemize}
\begin{algorithm}[ht]
    \caption{$\FindFrom{2}(l,r,s)$. Quantum Algorithm to search for any $\pm 2$ substring.\label{alg:substring2_base}}
    \begin{algorithmic}
        \State $i \gets  \textsc{Grover}(l,r, {\cal F}_s)$\Comment{Invoke Grover's search}
        \If{$i=-1$}
            \State \Return $\NULL$
        \EndIf
        \If{$i\neq-1$}
            \If{$x_i=x_{i+1}$}
                \State \Return $(i,i+1, f(x[i,i]))$
            \EndIf
            \If{$x_i=x_{i-1}$}
                \State \Return $(i-1,i, f(x[i,i]))$
            \EndIf
        \EndIf
    \end{algorithmic}
\end{algorithm}

\begin{lemma}
    The running time of Algorithm \ref{alg:substring2_base} is $O(\sqrt{n})$. The error probability is $O\left(1\right)$
\end{lemma}
\begin{proof}
    The main part of the algorithm is the Grover's Search algorithm that has $O(\sqrt{n})$ running time and $O\left(1\right)$ error probability.
\end{proof}

It will be useful to consider a modification of the algorithm that finds not just any $\pm 2$ substring, but the closest to the left border or to the right border. In that case, we use a subroutine $\textsc{Grover\_First\_One}$, with parameters $(l,r,{\cal F},direction)$ that accepts $l$ and $r$ as left and right borders of the search space, and a function ${\cal F}$ and a $direction\in\{left, right\}$.
\begin{itemize}
    \item  If $direction=left$, then we search for the maximal index $i$ such that ${\cal F}(i)=1$ where $l\leq i\leq r$.
    \item  If $direction=right$, then we search for the minimal index $i$ such that ${\cal F}(i)=1$ where $l\leq i\leq r$.
\end{itemize}
The result of  $\textsc{Grover\_First\_One}(l,r,{\cal F}, direction)$
is either $i$ or $-1$ if it has not found the required $i$.
See \cite{k2014,ll2016} on an implementation of such a function.

Algorithm \ref{alg:substring2_first} implements the $\FindFirst{2}$ subroutine. It has the same input and output parameters as $\FindFrom{2}$ and an extra input $direction\in\{left, right\}$. 

\begin{algorithm}[ht]
    \caption{$\FindFirst{2}(l,r,s, direction)$. 
    Searching for the first $\pm 2$-substring.\label{alg:substring2_first}}
    \begin{algorithmic}
        \State $i \gets  \textsc{Grover\_First\_One}(l,r, {\cal F}_s, direction)$\Comment{Invoke Grover's search}
        \If{$i=-1$}
            \State \Return $\NULL$
        \EndIf
        \If{$i\neq-1$}
            \If{$x_i=x_{i+1}$}
                \State \Return $(i,i+1, f(x[i,i]))$
            \EndIf
            \If{$x_i=x_{i-1}$}
                \State \Return $(i-1,i, f(x[i,i]))$
            \EndIf
        \EndIf
    \end{algorithmic}
\end{algorithm}

\begin{lemma}\label{lm:substring2_first}
    If the requested substring exists, the expected running time of Algorithm \ref{alg:substring2_first} is $O(\sqrt{j})$, where $j$ is the furthest border of the searching segment. Otherwise, the running time is $O(\sqrt{r-l})$. The error probability is at most $0.1$.
\end{lemma}
\begin{proof}
    The main part of the algorithm is $\textsc{Grover\_First\_One}$  \cite{k2014,ll2016} that has $O(\sqrt{j})$ expected running time and at most $0.1$ error probability. The running time is $O(\sqrt{r-l})$ if there are no $\pm2$-substrings.
\end{proof}